\documentclass{article}

\usepackage[utf8]{inputenc}
\usepackage[T1]{fontenc}
\usepackage{amsfonts}
\usepackage{amsmath}
\usepackage{amssymb}
\usepackage{amsthm}
\usepackage{enumerate}
\usepackage{lmodern}
\usepackage{url}

\newtheorem{theorem}{Theorem}
\newtheorem{corollary}[theorem]{Corollary}
\theoremstyle{definition}
\newtheorem{defn}[theorem]{Definition}{\bfseries}{\upshape}
\newtheorem{example}[theorem]{Example}

\newcommand{\N}{\mathbb{N}}
\newcommand{\Q}{\mathbb{Q}}
\newcommand{\R}{\mathbb{R}}
\newcommand{\Z}{\mathbb{Z}}
\newcommand{\sol}{z}
\newcommand\mbar[1]{\overline{#1}}

\let\originalleft\left
\let\originalright\right
\renewcommand{\left}{\mathopen{}\mathclose\bgroup\originalleft}
\renewcommand{\right}{\aftergroup\egroup\originalright}

\allowdisplaybreaks

\begin{document}

\author{Hoon Hong\\
North Carolina State University, Raleigh, NC, USA\\
\url{hong@ncsu.edu}
\and
Thomas Sturm\\
CNRS, Inria, and the University of Lorraine, Nancy, France\\
MPI Informatics and Saarland University, Saarbrücken, Germany\\
\url{thomas.sturm@loria.fr}}

\title{Positive Solutions of Systems of\\ Signed Parametric Polynomial Inequalities }
\maketitle

\begin{abstract}
  We consider systems of strict multivariate polynomial inequalities over the
  reals. All polynomial coefficients are parameters ranging over the reals,
  where for each coefficient we prescribe its sign. We are interested in the
  existence of positive real solutions of our system for all choices of
  coefficients subject to our sign conditions. We give a decision procedure for
  the existence of such solutions. In the positive case our procedure yields a
  parametric positive solution as a rational function in the coefficients. Our
  framework allows to reformulate heuristic subtropical approaches for
  non-parametric systems of polynomial inequalities that have been recently used
  in qualitative biological network analysis and, independently, in
  satisfiability modulo theory solving. We apply our results to characterize the
  incompleteness of those methods.
\end{abstract}

\section{Introduction}

We investigate the problem of finding a \emph{parametric positive} solution of a
system of \emph{signed parametric} polynomial inequalities, if exists. We
illustrate the problem by means of two toy examples:
\[
  f(x) =  c_2x^2-c_1x+c_0,\quad
  g(x) =  -c_2x^2+c_1x-c_0,
\]
where $c_2$, $c_1$, $c_0$ are parameters. An expression $\sol(c)$ is called a
parametric positive solution of $f(x)>0$ if for all $c>0$ we have $\sol(c)>0$
and $f(\sol(c))>0$. One easily verifies that $\sol(c)=\frac{c_1}{c_2}$ is a
parametric positive solution of $f(x)$. However, $g(x)>0$ does not have any
parametric positive solution since $g(x)>0$ has no positive solution when, e.g,
$c_2=c_1=c_0=1$. Of course, we are interested in tackling much larger cases with
respect to numbers of variables, monomials, and polynomials.

The problem is important as systems of polynomial inequalities often arise in
science and engineering applications, including, e.g., the qualitative analysis
of biological or chemical networks
\cite{WeberSturm:11a,ErramiEiswirth:15a,BradfordDavenport:17a,EnglandErrami:17b}
or Satisfiability Modulo Theories (SMT) solving
\cite{NieuwenhuisOliveras:06a,Abraham2016,FontaineOgawa:17b}. Surprisingly often
one is indeed interested in positive solutions. For instance, unknowns in the
biological and chemical context of
\cite{WeberSturm:11a,ErramiEiswirth:15a,BradfordDavenport:17a,EnglandErrami:17b}
are typically positive concentrations of species or reaction rates, where the
direction of the reaction is known. In SMT solving, positivity is often not
required but, in the satisfiable case, benchmarks typically have also positive
solutions; comprehensive statistical data for several thousand benchmarks can be
found in \cite[Sect.~6]{FontaineOgawa:17b}. In many areas systems have
parameters and one desires to have parametric solutions. Hence, an efficient and
reliable tool for finding parametric positive solutions can aid scientists and
engineers in developing and investigating their mathematical models.

The problem of finding parametric positive solution is essentially that of
quantifier elimination over the first order theory of real closed fields. In
1930, Alfred Tarski~\cite{Tarski:30} showed that real quantifier elimination can
be carried out algorithmically. Since then, there have been intensive research,
producing profound theories, dramatically improved algorithms, and highly
refined implementations in widely available computer algebra software such as
Mathematica, Maple, Qepcad~B, or Reduce, e.g.,
\cite{Tarski:30,Collins:75,Arnon:81,McCallum:84,Grigorev_Vorobjov:88,%
  Canny:88,Hong:90a,Hong:90b,Collins_Hong:91,Weispfenning:95,SturmW96,%
  Sturm:99a,DBLP:conf/issac/Brown00,Strzebonski06,Hong_Safey:2012,Kosta:16a}.
However, existing general quantifier elimination algorithms are still too
inefficient for tackling even small problems of finding parametric positive
solutions.

The main contribution of this paper is to provide simple and efficient
algorithmic criteria for deciding whether or not a given signed parametric
system has a parametric positive solution. If so, we provide an explicit formula (rational function) 
for a parametric positive solution.
%
The main challenge was eliminating many universal quantifiers in the problem
statement. We tackled that challenge by, firstly, carefully
approximating/bounding polynomials by suitable multiple of monomials and,
secondly, tropicalizing, i.e., linearizing monomials by taking logarithms in the
style of \cite{2000math......5163V}. However, unlike standard tropicalization
approaches, we determine sufficiently large \emph{finite bases} for our
logarithms, in order to get an explicit formula for parametric positive
solutions.

Our main result also shines a new light on recent heuristic subtropical
methods~\cite{Sturm:15b,FontaineOgawa:17b}: We provide a precise
characterization of their incompleteness in terms of the existence of parametric
positive solutions for the originally non-parametric input problems considered
there.

The paper is structured as follows. In Section~\ref{sec:notation}, we motivate
and present a compact and convenient notation for a systems of multivariate
polynomials, which will be used throughout the paper. In Section~\ref{sec:main},
we precisely define the key notions of \emph{signed parametric systems} and
\emph{parametric positive solutions}. Then we present and prove the main result
of this paper, which shows how to check the existence of a parametric positive
solution and, in the positive case, how to find one. In
Section~\ref{sec:subtropical}, we apply use our framework and our result to
re-analyze and improve the above-mentioned subtropical
methods~\cite{Sturm:15b,FontaineOgawa:17b}.

\section{Notation}\label{sec:notation}

The principal mathematical object studied in this paper are systems of
multivariate polynomials over the real numbers. In order to minimize cumbersome
indices, we are going to introduce some compact notations. Let us start with a
motivation by means of a simple example.%
\begin{example}
  Consider the following system of three polynomials in two variables:
  \begin{align*}
    f_{1} & =  2x_{1}^{2}x_{2}-x_{1}^{3}\\
    f_{2} & =  -3x_{1}x_{2}^{2}+6x_{1}^{3}\\
    f_{3} & =  -x_{1}^{2}+5x_{1}^{1}x_{2}^{2}.
  \end{align*}
  We rewrite those polynomials by aligning their signs, coefficients, and
  monomials supports:%
  \[%
    \begin{array}[c]{c@{\quad}c@{\quad}rr@{\quad}c@{\quad}rr@{\quad}c@{\quad}rr}%
      f_{1} & = & 1\cdot2\cdot{} & x_{1}^{2}x_{2}^{1} &  +  & 0\cdot1\cdot{} &
                  x_{1}^{1}x_{2}^{2} & + & -1\cdot4\cdot{}& x_{1}^{3}x_{2}^{0}\\
      f_{2} & = & 0\cdot1\cdot{} & x_{1}^{2}x_{2}^{1} &  +  & -1\cdot3\cdot{} &
                 x_{1}^{1}x_{2}^{2} & + & 1\cdot6\cdot{} & x_{1}^{3}x_{2}^{0}\\
      f_{3} & = & -1\cdot1\cdot{} & x_{1}^{2}x_{2}^{1}&  +  & 1\cdot5\cdot{} &
                 x_{1}^{1}x_{2}^{2} & + & 0\cdot1\cdot{} & x_{1}^{3}x_{2}^{0}\rlap{,}
    \end{array}
  \]
  where signs are represented by $-1$, $0$, and $1$. Note that we are writing
  $0$ coefficients as $0\cdot1$ for notational uniformity. Rewriting this in
  matrix-vector notation, we have%
  \[
    \left[
      \begin{array} [c]{c}%
        f_{1}\\
        f_{2}\\
        f_{3}%
      \end{array}
    \right] =\left( \left[
        \begin{array} [c]{rrr}%
          1 & 0 & -1\\
          0 & -1 & 1\\
          -1 & 1 & 0
        \end{array}
      \right] \circ\left[
        \begin{array} [c]{rrr}%
          2 & 1 & 4\\
          1 & 3 & 6\\
          1 & 5 & 1
        \end{array}
      \right] \right) \left[
      \begin{array} [c]{c}%
        x_{1}^{2}x_{2}^{1}\\
        x_{1}^{1}x_{2}^{2}\\
        x_{1}^{3}x_{2}^{0}%
      \end{array}
    \right],
  \]
  where $\circ$ is the component-wise Hadamard product. Pushing this even
  further, we have%
  \[
    \left[
      \begin{array} [c]{c}%
        f_{1}\\
        f_{2}\\
        f_{3}%
      \end{array}
    \right] =\left( \left[
        \begin{array} [c]{rrr}%
          1 & 0 & -1\\
          0 & -1 & 1\\
          -1 & 1 & 0
        \end{array}
      \right] \circ\left[
        \begin{array} [c]{rrr}%
          2 & 1 & 4\\
          1 & 3 & 6\\
          1 & 5 & 1
        \end{array}
      \right] \right) \left[
      \begin{array} [c]{cc}%
        x_{1} & x_{2}%
      \end{array}
    \right] ^{\left[
        \begin{array} [c]{cc}%
          2 & 1\\
          1 & 2\\
          3 & 0
        \end{array}
      \right]}%
    .\]
  Thus we have arrived at a form
  \[
    f=\left( s\circ c\right) x^{e},
  \]
  where
  \begin{multline*}
    f=\left[
      \begin{array} [c]{c}%
        f_{1}\\
        f_{2}\\
        f_{3}%
      \end{array}
    \right],\quad
    s=\left[
      \begin{array} [c]{rrr}%
        1 & 0 & -1\\
        0 & -1 & 1\\
        -1 & 1 & 0
      \end{array}
    \right],\quad
    c=\left[
      \begin{array} [c]{rrr}%
        2 & 1 & 4\\
        1 & 3 & 6\\
        1 & 5 & 1
      \end{array}
    \right],\\
    x=[%
    \begin{array} [c]{cc}%
      x_{1} & x_{2}%
    \end{array}
    ],\quad
    e=\left[
      \begin{array} [c]{c}%
        e_{1}\\
        e_{2}\\
        e_{3}%
      \end{array}
    \right] =\left[
      \begin{array} [c]{cc}%
        2 & 1\\
        1 & 2\\
        3 & 0
      \end{array}
    \right].
\end{multline*}
This concludes our example.
\end{example}

In general, a system $f\in\R[x_1,\dots,x_d]^u$ of multivariate polynomials over
the reals will be written compactly as
\[
f=\left(  s\circ c\right)  x^{e},
\]
where
\begin{multline*}
  f=\left[
    \begin{array}
      [c]{c}%
      f_{1}\\
      \vdots\\
      f_{u}%
    \end{array}
  \right],\quad
  s=\left[
    \begin{array}
      [c]{ccc}%
      s_{11} & \cdots & s_{1v}\\
      \vdots &  & \vdots\\
      s_{u1} & \cdots & s_{uv}%
    \end{array}
  \right],\quad
  c=\left[
    \begin{array}
      [c]{ccc}%
      c_{11} & \cdots & c_{1v}\\
      \vdots &  & \vdots\\
      c_{u1} & \cdots & c_{uv}%
    \end{array}
  \right],\\
  x=[%
  \begin{array}
    [c]{ccc}%
    x_{1} & \cdots & x_{d}%
  \end{array}
  ],\quad
  e=\left[
    \begin{array}
      [c]{c}%
      e_{1}\\
      \vdots\\
      e_{v}%
    \end{array}
  \right]  =\left[
    \begin{array}
      [c]{ccc}%
      e_{11} & \cdots & e_{1d}\\
      \vdots &  & \\
      e_{v1} & \cdots & e_{vd}%
    \end{array}
  \right].
\end{multline*}
We call $s\in\left\{ -1,0,1\right\}^{u\times v}$ the \emph{sign matrix},
$c\in\mathbb{R}_{+}^{u\times v}$ the \emph{coefficient matrix}, and
$e\in\mathbb{N}^{v\times d}$ the \emph{exponent matrix} of $f$.

\section{Main Result}\label{sec:main}

\begin{defn} [Signed parametric systems]
  A \emph{signed parametric system} is given by
  \[
    f=\left(  s\circ c\right)  x^{e},
  \]
  where the sign matrix $s\in\{-1,0,1\}^{u\times v}$ and the exponent matrix
  $e\in\N^{v\times d}$ are specified but the coefficient matrix $c$ is
  unspecified in the sense that it is left parametric. Formally $c$ is a
  $u\times v$-matrix of pairwise different indeterminates.
\end{defn}

\begin{defn}[Parametric positive solutions]
  Consider a signed parametric system $f=\left( s\circ c\right) x^{e}$. A
  \emph{parametric positive solution} of $f(x)>0$ is a function
  $\sol:\R_+^{u\times v}\to\R_+^d$ that maps each possible specification of the
  coefficient matrix $c$ to a solution of the corresponding non-parametric
  system, i.e.,
  \[
    \underset{c>0}{\forall}\ f\bigl(\sol(c)\bigr)>0.
  \]
\end{defn}

\newpage
\begin{theorem}[Main]\label{TH:main}
Let $f=\left(  s\circ c\right)  x^{e}$ be a signed parametric system.
Let
\[
  C(n)\ :=\ \bigwedge_{i}\ \bigwedge_{s_{ik}<0}\ \bigvee_{s_{ij}>0}\
  (e_{j}-e_{k})\cdot n\geq1.
\]
Then the following are equivalent:
\begin{enumerate}[(i)]
\item $f(x)>0$ has a parametric positive solution.
\item $C(n)$ has a solution $n\in\mathbb{R}^{d}$.
\item $C(n)$ has a solution $n\in\mathbb{Z}^{d}$.
\end{enumerate}
In the positive case, the following function $\sol$ is a parametric positive
solution of $f(x)>0$:
\[
  \sol(c)=t^{n},\quad
  \text{where}\quad
  t=1+\sum_{s_{ij}>0\atop s_{ik}<0}\frac{c_{ik}}{c_{ij}}.
\]
In fact, we even have
$\underset{c>0}{\forall}\ \underset{r\geq t}{\forall}\ f(r^n)>0$.
\end{theorem}

\begin{proof}
  We first show that (i) implies (ii):
  \begin{align*}
    \text{(i)} \ \
    \Longleftrightarrow
    & \ \ \underset{c>0}{\forall}
      \ \ \underset{x>0}{\exists}
      \ \ \left(s\circ c\right)  x^{e}>0\\
    \Longleftrightarrow
    & \ \ \underset{c>0}{\forall}
      \ \ \underset{x>0}{\exists}
      \ \ \bigwedge_{i}
      \ \ \sum_{s_{ij}>0}c_{ij}x^{e_{j}}>\sum_{s_{ik}<0}
      c_{ik}x^{e_{k}}\\
    \Longrightarrow
    & \ \ \underset{x>0}{\exists}
      \ \ \bigwedge_{i}
      \ \ \sum_{s_{ij}>0}x^{e_{j}}>\sum_{s_{ik}<0}2vx^{e_{k}},\quad
      \text{by instantiating $c$}\\
    \Longrightarrow
    & \ \ \underset{x>0}{\exists}
      \ \ \bigwedge_{i}
      \ \ v\max_{s_{ij}>0}x^{e_{j}}>\max_{s_{ik}>0}2vx^{e_{k}}\\
    \Longleftrightarrow
    & \ \ \underset{x>0}{\exists}
      \ \ \bigwedge_{i}
      \ \ \max_{s_{ij}>0}x^{e_{j}}>\max_{s_{ik}>0}2x^{e_{k}}\\
    \Longleftrightarrow
    & \ \ \underset{x>0}{\exists}
      \ \ \bigwedge_{i}
      \ \ \bigwedge_{s_{ik}<0}
      \ \ \bigvee_{s_{ij}>0}
      \ \ x^{e_{j}}>2x^{e_{k}}\\
    \Longleftrightarrow
    & \ \ \underset{x>0}{\exists}
      \ \ \bigwedge_{i}
      \ \ \bigwedge_{s_{ik}<0}
      \ \ \bigvee_{s_{ij}>0}
      \ \ x^{e_{j}-e_{k}}>2\\
    \Longleftrightarrow
    & \ \ \underset{x>0}{\exists}
      \ \ \bigwedge_{i}
      \ \ \bigwedge_{s_{ik}<0}
      \ \ \bigvee_{s_{ij}>0}
      \ \ \left(e_{j}-e_{k}\right)\cdot\log_{2}x>1\\
    \Longleftrightarrow
    & \ \ \underset{n\in\R^d}{\exists}
      \ \ \bigwedge_{i}
      \ \ \bigwedge_{s_{ik}<0}
      \ \ \bigvee_{s_{ij}>0}
      \ \ \left(  e_{j}-e_{k}\right)\cdot n>1,\quad
      \text{using $\log_{2}:\R_+\leftrightarrow\R$}\\
    \Longrightarrow
    & \ \ \text{(ii)}.
  \end{align*}

  Assume now that (ii) holds. The existence of solutions $n\in\R^d$ and
  $n\in\Q^d$ of $C(n)$ coincides due to the Linear Tarski Principle: Ordered
  fields admit quantifier elimination for linear formulas, and therefore $\Q$ is
  an elementary substructure of $\R$ with respect to linear sentences
  \cite{LoosWeispfenning:93a}. Given a solution $n\in\Q^d$, we can use the
  principal denominator $\delta>0$ of all coordinates of $n$ to obtain a
  solution $\delta n\in\Z^d$. Hence (iii) holds.

  We finally  show that (iii) implies (i):
  \begin{align*}
    \text{(i)}\ \
    \Longleftrightarrow
    & \ \ \underset{c>0}{\forall}
      \ \ \underset{x>0}{\exists}
      \ \ \left(  s\circ c\right)  x^{e}>0\\
    \Longleftrightarrow
    &  \ \ \underset{c>0}{\forall}
      \ \ \underset{x>0}{\exists}
      \ \ \bigwedge_{i}
      \ \ \sum_{s_{ij}>0}c_{ij}x^{e_{j}}>\sum_{s_{ik}<0}c_{ik}x^{e_{k}}\\
    \Longleftarrow
    &  \ \ \underset{c>0}{\forall}
      \ \ \underset{x>0}{\exists}
      \ \ \bigwedge_{i}
      \ \ \max_{s_{ij}>0}c_{ij}x^{e_{j}}>\mbar{c}_{i}\max_{s_{ik}<0}x^{e_{k}},\quad
      \text{where $\mbar{c}_{i}=\sum_{s_{ik'}<0}c_{ik'}$}\\
    \Longleftrightarrow
    & \ \ \underset{c>0}{\forall}
      \ \ \underset{x>0}{\exists}
      \ \ \bigwedge_{i}
      \ \ \bigwedge_{s_{ik}<0}
      \ \ \bigvee_{s_{ij}>0}
      \ \ c_{ij}x^{e_{j}}>\mbar{c}_{i}x^{e_{k}}\\
    \Longleftrightarrow
    & \ \ \underset{c>0}{\forall}
      \ \ \underset{x>0}{\exists}
      \ \ \bigwedge_{i}
      \ \ \bigwedge_{s_{ik}<0}
      \ \ \bigvee_{s_{ij}>0}
      \ \ x^{e_{j}-e_{k}}>\frac{\mbar{c}_{i}}{c_{ij}}\\
    \Longleftarrow
    & \ \ \underset{c>0}{\forall}
      \ \ \underset{x>0}{\exists}
      \ \ \bigwedge_{i}
      \ \ \bigwedge_{s_{ik}<0}
      \ \ \bigvee_{s_{ij}>0}
      \ \ x^{e_{j}-e_{k}}\geq t,\\
    & \ \ \text{where $t$ is as stated in the theorem}\\
    \Longleftrightarrow
    &  \ \ \underset{c>0}{\forall}
      \ \ \underset{x>0}{\exists}
      \ \ \bigwedge_{i}
      \ \ \bigwedge_{s_{ik}<0}
      \ \ \bigvee_{s_{ij}>0}
      \ \ \left(e_{j}-e_{k}\right)\cdot\log_{t}x\geq1,\quad\\
    \Longleftrightarrow
    & \ \ \underset{n}{\exists}
      \ \ \bigwedge_{i}
      \ \ \bigwedge_{s_{ik}<0}
      \ \ \bigvee_{s_{ij}>0}
      \ \ \left(  e_{j}-e_{k}\right)\cdot n\geq1,\quad
      \text{using $\log_{t}:\mathbb{R}_{+}\leftrightarrow \mathbb{R}$}\\
    \Longleftarrow
    & \ \ \underset{n\in\mathbb{Z}^{d}}{\exists}
      \ \ \bigwedge_{i}
      \ \ \bigwedge_{s_{ik}<0}
      \ \ \bigvee_{s_{ij}>0}
      \ \ \left(e_{j}-e_{k}\right)\cdot n\geq1\\
    \Longleftrightarrow
    & \ \ \text{(iii)}.
  \end{align*}

  In our proof of the implication from (iii) to (i) we have applied $\log_t$ so
  that $n=\log_t x$ and, accordingly, $x=t^n$, where $t$ is as stated in the
  theorem. Notice that any larger choice $r\geq t$ would work there as well.
\end{proof}

\begin{example}\label{ex:hassol}
  Consider $\scriptsize f=\left[
  \begin{array} [c]{c}%
    f_{1}\\
    f_{2}
  \end{array}\right]$
  with
  \begin{align*}
    f_1&=-c_{11}x_1^5+c_{12}x_1^2x_2-c_{13}x_1^2+c_{15}x_2^2\\
    f_2&=c_{21}x_1^5+c_{22}x_1^2x_2+c_{23}x_1^2-c_{24}x_2^3.
  \end{align*}
  That is
  \[
    s=\left[
      \begin{array}
        [c]{ccccc}%
        -1 & 1 & -1 & 0 & 1\\
        1 & 1 & 1 & -1 & 0
      \end{array}
    \right],\quad
    e=\left[
      \begin{array}
        [c]{cc}%
        5 & 0\\
        2 & 1\\
        2 & 0\\
        0 & 3\\
        0 & 2
      \end{array}
    \right].
  \]
  Then $C(n)$ has a solution $n\in\mathbb{Z}^{2}$, e.g.,
  \[
    n=\left[
      \begin{array}
        [c]{cc}%
        -12\  & -11
      \end{array}
    \right].
  \]
  Hence $f=(s\circ c)  x^{e}>0$ has a parametric positive solution, e.g.,
  \begin{displaymath}
    \sol(c) = \left[
      \begin{array}
        [c]{cc}%
        t^{-12}\  & t^{-11}
      \end{array}\right],
  \end{displaymath}
  where
  $\displaystyle t=1+\frac{c_{11}}{c_{12}}+
  \frac{c_{11}}{c_{15}}+\frac{c_{13}}{c_{12}}+\frac{c_{13}}{c_{15}}+\frac{c_{24}}{c_{21}}+
  \frac{c_{24}}{c_{22}}+\frac{c_{24}}{c_{23}}$.
\end{example}

\begin{example}
  We slightly modify Example~\ref{ex:hassol} and consider
  $\scriptsize f=\left[
    \begin{array}[c]{c}%
      f_{1}\\
      f_{2}
    \end{array}\right]$
  with
  \begin{align*}
    f_1&=-c_{11}x_1^5+c_{12}x_1^2x_2-c_{13}x_1^2+c_{15}x_2^2\\
    f_2&=c_{21}x_1^5+c_{22}x_1^2x_2-c_{23}x_1^2-c_{24}x_2^3.
  \end{align*}
  That is
  \[
    s=\left[
      \begin{array}
        [c]{ccccc}%
        -1 & 1 & -1 & 0 & 1\\
        1 & 1 & -1 & -1 & 0
      \end{array}
    \right],\quad
    e=\left[
      \begin{array}
        [c]{cc}%
        5 & 0\\
        2 & 1\\
        2 & 0\\
        0 & 3\\
        0 & 2
      \end{array}
    \right].
  \]
  Then $C(n)$ does not have a solution $n\in\mathbb{Z}^{2}$. Hence
  $f=( s\circ c) x^{e} > 0$ does not have a parametric positive
  solution.
\end{example}

\section{A Re-analysis of Subtropical Methods}\label{sec:subtropical}
For non-parametric systems of real polynomial inequalities, heuristic Newton
polytope-based \emph{subtropical methods} \cite{Sturm:15b,FontaineOgawa:17b}
have been successfully applied in two quite different areas: Firstly,
qualitative analysis of biological and chemical networks and, secondly, SMT
solving.

In the first area, a positive solution of a very large single inequality could
be computed. The left hand side polynomial there has more than $8\cdot10^5$
monomials in 10 variables with individual degrees up to 10. This computation was
the hard step in finding an exact positive solution of the corresponding
equation using a known positive point with negative value of the polynomial and
applying the intermediate value theorem. To give a very rough idea of the
biological background: The polynomial is a Hurwitz determinant originating from
a system of ordinary differential equations modeling mitogen-activated protein
kinase (MAPK) in the metabolism of a frog. Positive zeros of the Hurwitz
determinant point at Hopf bifurcations, which are in turn indicators for
possible oscillation of the corresponding reaction network. For further details
see \cite{ErramiEiswirth:15a}.

In the second area, a subtropical approach for systems of several polynomial
inequalities has been integrated with an SMT solver. That combination could
solve a surprisingly large percentage of SMT benchmarks very fast and thus
establishes an interesting heuristic preprocessing step for the SMT theory of
QF\_NRA (quantifier-free nonlinear arithmetic). For detailed statistics see
\cite[Sect.~6]{FontaineOgawa:17b}.

 The goal of this section is, to make precise the connections between subtropical
methods and our main result here, to use these connections to improve the
subtropical methods, and to precisely characterize their incompleteness.

\subsection{Subtropical Real Root Finding}
In \cite{Sturm:15b} we have studied an incomplete method for heuristically
finding a positive solution for a single multivariate polynomial inequality with
fixed integer coefficients:
\begin{displaymath}
  [f_1] = (s\circ c)x^{e}\quad
  \text{where}\quad
  s\in\{-1,0,1\}^{1\times v},\quad
  c\in\Z_+^{1\times v},\quad
  e\in\N^{v\times d}.
\end{displaymath}
The method considers the positive and the negative support, which in terms of
our notions is given by
\begin{displaymath}
  S^+=\{\,e_j\mid s_{1j}>0\,\},\quad
  S^-=\{\,e_k\mid s_{1k}<0\,\}.
\end{displaymath}
Then \cite[Lemma 4]{Sturm:15b} essentially states that $f_1(x)>0$ has a positive
solution if
\begin{multline*}
  C'\ :=
  \bigvee_{e_j\in S^+}
  \ \underset{n\in\R^d}{\exists}
  \ \underset{\gamma\in\R}{\exists}
  \ \Biggl(
  \left[
    \begin{array}{cc}
      -e_j & 1
    \end{array}
  \right]
  \left[
    \begin{array}{c}
      n\\
      \gamma
    \end{array}
  \right]\leq -1
  \land{}\\
  \bigwedge_{e_k\in S^+\cup S^{-}\atop e_k\neq e_j}
  \left[
    \begin{array}{cc}
      e_k & -1
    \end{array}
  \right]
  \left[
    \begin{array}{c}
      n\\
      \gamma
    \end{array}
  \right]\leq -1
  \Biggr).
\end{multline*}
Unfortunately, in \cite[Lemma~4]{Sturm:15b} vectors
$e_l=\left[\begin{array}{ccc} 0&\dots&0\end{array}\right]$ corresponding to
absolute summands are treated specially. We have noted already in
\cite[p.192]{FontaineOgawa:17b} that an inspection of the proof shows that this
is not necessary. Therefore we discuss here a slightly improved and simpler
version without that special treatment, which has been explicitly stated as
\cite[Lemma~2]{FontaineOgawa:17b}.

The proof of the loop invariant (I\textsubscript{1}) in
\cite[Theorem~5(ii)]{Sturm:15b} shows that the positive support need not be
considered in the conjunction:
\begin{multline*}
  C'\ \Longleftrightarrow
  \bigvee_{e_j\in S^+}
  \ \underset{n\in\R^d}{\exists}
  \ \underset{\gamma\in\R}{\exists}
  \ \Biggl(
    \left[
      \begin{array}{cc}
        -e_j & 1
      \end{array}
    \right]
    \left[
      \begin{array}{c}
        n\\
        \gamma
      \end{array}
    \right]\leq 1
    \land{}\\
    \bigwedge_{e_k\in S^{-}}
    \left[
      \begin{array}{cc}
        e_k & -1
      \end{array}
    \right]
    \left[
      \begin{array}{c}
        n\\
        \gamma
      \end{array}
    \right]\leq -1
  \Biggr).
\end{multline*}
Starting with Fourier--Motzkin elimination \cite[Sect.~12.2]{Schrijver86} of
$\gamma$, we obtain
\begin{align*}
  C'\ \
  & \Longleftrightarrow
    \ \ \bigvee_{e_j\in S^+}
    \ \ \underset{n\in\R^d}{\exists}
    \ \ \bigwedge_{e_k\in S^{-}}
    \ \ (e_k-e_j)\cdot n \leq -2\\
  & \Longleftrightarrow
    \ \ \bigvee_{e_j\in S^+}
    \ \ \underset{n\in\R^d}{\exists}
    \ \ \bigwedge_{e_k\in S^{-}}
    \ \ (e_j-e_k)\cdot n \geq 1\\
  & \Longleftrightarrow
    \ \ \underset{n\in\R^d}{\exists}
    \ \ \bigvee_{e_j\in S^+}
    \ \ \bigwedge_{e_k\in S^{-}}
    \ \ (e_j-e_k)\cdot n \geq 1\\
  & \Longleftrightarrow
    \ \ \underset{n\in\R^d}{\exists}
    \ \ \max_{e_j\in S^+}(e_j\cdot n) \geq \max_{e_k\in S^-}(e_k \cdot n+1)\\
  & \Longleftrightarrow
    \ \ \underset{n\in\R^d}{\exists}
    \ \ \bigwedge_{e_k\in S^{-}}
    \ \ \bigvee_{e_j\in S^+}
    \ \ (e_j-e_k)\cdot n \geq 1\\
  & \Longleftrightarrow
    \ \ \underset{n\in\R^d}{\exists}
    C(n)
\end{align*}
with $C(n)$ as in Theorem~\ref{TH:main}.

\begin{corollary}\label{co:issac}
  Consider $f_1\in\Z[x_1,\dots,x_d]$, say, $f_1=(s\circ c)x^e$, where
  $s\in\{-1,0,1\}^{1\times v}$, $c\in\Z_+^{1\times v}$, $e\in\N^{v\times d}$.
  Let $f^*=[(s\circ c^*)x^e]$, where $c^*$ is a $1\times v$-matrix of pairwise
  different indeterminates. Then the following are equivalent:
  \begin{enumerate}[(i)]
  \item The algorithm \texttt{\upshape find-positive}
    \cite[Algorithm~1]{Sturm:15b} does not fail, and thus finds a rational
    solution of $f_1>0$ with positive coordinates.
  \item There is a row $e_j$ of $e$ with $s_{1j}>0$ such that the following LP
    problem has a solution $n\in\Q^d$:
    \begin{displaymath}
      \bigwedge_{s_{1k}<0} (e_j-e_k)\cdot n \geq 1.
    \end{displaymath}
  \item $f^*>0$ has a parametric positive solution.
  \end{enumerate}
  In the positive case, $f(r^n)>0$ for all
  $r\geq 1+v\sum\limits_{s_{1k}<0}c_{1k}$.
\end{corollary}

\begin{proof}
  The equivalence between (i), (ii), and (iii) has been derived above.

  According to Theorem~\ref{TH:main}, a solution for $f_1>0$ can be obtained by
  plugging~$c$ into the parametric positive solution for $f^*$. Since we have
  positive integer coefficients, we can bound $t$  from above as follows.
  \begin{displaymath}
    t=1+\sum\limits_{s_{1j}>0\atop s_{1k}<0}\frac{c_{1k}}{c_{1j}}
    \leq
    1+\sum\limits_{s_{1j}>0\atop s_{1k}<0}\frac{c_{1k}}{1}
    \leq
    1+v\sum\limits_{s_{1k}<0}c_{1k}.\qedhere
  \end{displaymath}
\end{proof}

In simple words the equivalence between (i) and (iii) in the corollary states
the following: The incomplete heuristic \cite[Algorithm~1]{Sturm:15b} succeeds
\emph{if and only if} not only the inequality for the input polynomial has a
solution as required, but also the inequality for all polynomials with the same
monomials and signs of coefficients as the input polynomial.

We have added (ii) to the corollary, because we consider this form optimal for
algorithmic purposes. Our special case of one single inequality allows to
transform the conjunctive normal form provided by Theorem~\ref{TH:main} into an
equivalent disjunctive normal form without increasing size. This way, a decision
procedure can use finitely many LP solving steps \cite{Schrijver86} instead of
employing more general methods like SMT solving \cite{NieuwenhuisOliveras:06a}.

Finally notice that the brute force search for a suitable $t$ in
\texttt{find-positive} \cite[Algorithm~1, l.10--12]{Sturm:15b} is not necessary
anymore. Our corollary computes a suitable number from the coefficients.

\subsection{Subtropical Satisfiability Checking}
Subsequent work \cite{FontaineOgawa:17b} takes an entirely geometric approach to
generalize the work in \cite{Sturm:15b} from one polynomial inequality to
finitely many such inequalities. Consider a system with fixed integer
coefficients in our notation:
\begin{displaymath}
  f=\left[
    \begin{array}
      [c]{c}%
      f_{1}\\
      \vdots\\
      f_{u}%
    \end{array}
    \right]
    =(s\circ c)x^e,\quad
    \text{where}\quad
    s\in\{-1,0,1\}^{u\times v},\quad
    c\in\Z_+^{u\times v},\quad
    e\in\N^{v\times d}.
\end{displaymath}
Then \cite[Theorem~12]{FontaineOgawa:17b} derives the following sufficient
condition for the existence of a positive solution of $f>0$:
\begin{multline*}
  C'' \ := 
  \underset{n\in\R^d}{\exists}
  \ \underset{\gamma_1\in\R}{\exists}\dots\underset{\gamma_u\in\R}{\exists}
  \ \bigwedge_{i=1}^u
  \ \Biggl(
    \Biggl(
      \bigvee_{s_{ij}>0}ne_{j}+\gamma_i>0
    \Biggr)
    \land
    \bigwedge_{s_{ik}<0}ne_{k}+\gamma_i<0
  \Biggr).
\end{multline*}
After an equivalence transformation, we can once more apply Fourier--Motzkin
elimination \cite[Sect.~12.2]{Schrijver86}:
\begin{align*}
  C'' \ \
  &\Longleftrightarrow
    \ \ \underset{n\in\R^d}{\exists}
    \ \ \bigwedge_{i=1}^u
    \ \ \bigvee_{s_{ij}>0}
    \ \ \underset{\gamma_i\in\R}{\exists}
    \ \ \left(
    ne_{j}+\gamma_i>0
    \land
    \bigwedge_{s_{ik}<0}ne_{k}+\gamma_i<0
    \right)\\
  &\Longleftrightarrow
    \ \ \underset{n\in\R^d}{\exists}
    \ \ \bigwedge_{i=1}^u
    \ \ \bigvee_{s_{ij}>0}
    \ \ \bigwedge_{s_{ik}<0}
    \ \ (e_{j}-e_k)\cdot n>0\\
  & \Longleftrightarrow
    \ \ \underset{n\in\R^d}{\exists}
    \ \ \bigwedge_{i=1}^u
    \ \ \max_{s_{ij}>0}e_jn > \max_{s_{ik}<0}e_kn\\
  &\Longleftrightarrow
    \ \ \underset{n\in\R^d}{\exists}
    \ \ \bigwedge_{i=1}^u
    \ \ \bigwedge_{s_{ik}<0}
    \ \ \bigvee_{s_{ij}>0}
    (e_{j}-e_k)\cdot n>0\\
  &\Longleftrightarrow
    \ \ \underset{n\in\R^d}{\exists}
    \ \ \bigwedge_{i=1}^u
    \ \ \bigwedge_{s_{ik}<0}
    \ \ \bigvee_{s_{ij}>0}
    \ \ (e_{j}-e_k)\cdot n\geq1\\
  &\Longleftrightarrow
    \ \ \underset{n\in\R^d}{\exists}
    C(n)
\end{align*}
with $C(n)$ as in Theorem~\ref{TH:main}.
\begin{corollary}
  Consider $f\in\Z[x_1,\dots,x_d]^u$, say, $f=(s\circ c)x^e$, where
  $s\in\{-1,0,1\}^{u\times v}$, $c\in\Z_+^{u\times v}$, $e\in\N^{v\times d}$.
  Let $f^*=(s\circ c^*)x^e$, where $c^*$ is a $u\times v$-matrix of pairwise
  different indeterminates. Then the following are equivalent:
  \begin{enumerate}[(i)]
  \item The incomplete subtropical satisfiability checking method for several
    inequalities over QF\_NRA (quantifier-free nonlinear real arithmetic)
    introduced in \cite{FontaineOgawa:17b} succeeds on $f>0$.
  \item The following SMT problem with unknowns $n$ is satisfiable in QF\_LRA
    (quantifier-free linear real arithmetic):
    \begin{displaymath}
      \bigwedge_{i=1}^u
      \ \bigwedge_{s_{ik}<0}
      \ \bigvee_{s_{ij}>0}
      \ (e_{j}-e_k)\cdot n\geq1.
  \end{displaymath}
  \item $f^*>0$ has a parametric positive solution.
  \end{enumerate}
  In the positive case, $f(r^n)>0$ for all
  $r\geq 1+v\sum\limits_{s_{ik}<0}c_{ik}$.
\end{corollary}

\begin{proof}
  The equivalence between (i), (ii), and (iii) has been derived above. About the
  solution $r$ see the proof of Corollary~\ref{co:issac}.
\end{proof}

The equivalence between (i) and (iii) in the corollary states the following: The
procedure in \cite{FontaineOgawa:17b} yields ``sat'' in contrast to ``unknown''
\emph{if and only if} not only the input system is satisfiable, but that system
with all real choices of coefficients with the same signs as in the input
system. While there are no formal algorithms in \cite{FontaineOgawa:17b}, the
work has been implemented within a combination of the veriT solver
\cite{Bouton:2009:VOT:1614530.1614546} with the library STROPSAT
\cite[Sect.~6]{FontaineOgawa:17b}. Our characterization applies in particular to
the completeness of this software.

We have added (ii) to the corollary, because we consider this form optimal for
algorithmic purposes. Like the original input $C''$ used in
\cite{FontaineOgawa:17b} this is a conjunctive normal form, which is ideal for
DPLL-based SMT solvers \cite{NieuwenhuisOliveras:06a}. Recall that $u$ is the
number of inequalities in the input, and $d$ is the number of variables. Let
$\iota$ and $\kappa$ be the numbers of positive and negative coefficients,
respectively. Then compared to \cite{FontaineOgawa:17b} we have reduced $d+u$
variables to $d$ variables, and we have reduced $u\kappa$ clauses with $\iota$
atoms each plus $u$ unit clauses to some different $u\kappa$ clauses with
$\iota$ atoms each but without any additional unit clauses.

With the \texttt{:produce-model} option the SMT-LIB standard
\cite{BarrettFontaine:17a} supports solutions like the $r^n$ provided by our
corollary. The work in \cite{FontaineOgawa:17b} does not address the computation
of solutions. It only mentions that sufficiently large $r$ will work, which
implicitly suggests a brute-force search like the one in \cite[Algorithm~1,
l.10--12]{Sturm:15b}.

\subsubsection*{Acknowledgments}
This work has been supported by the European Union's Horizon 2020 research and
innovation programme under grant agreement No H2020-FETOPEN-2015-CSA 712689
SC-SQUARE and by the bilateral project ANR-17-CE40-0036 and DFG-391322026
SYMBIONT. The second author would like to thank Georg Regensburger for his
hospitality and an interesting week of inspiring discussions around the topic,
and Dima Grigoriev for getting him started on the subject.

\end{document}